\documentclass[aps,pra,superscriptaddress,floatfix,showpacs,notitlepage,twocolumn]{revtex4-1}

\usepackage[latin9]{inputenc}
\usepackage{color}
\usepackage{textcomp}
\usepackage{amsmath}
\usepackage{amssymb}
\usepackage{amsthm}
\usepackage{enumitem}
\usepackage{graphicx}
\usepackage{esint}
\usepackage{framed} 
\usepackage{xcolor}
\usepackage{dsfont}
\usepackage{hyperref}
\usepackage{nccmath}
\hypersetup{colorlinks,linkcolor=blue,citecolor=blue,filecolor=blue,urlcolor=blue}
\makeatletter
\usepackage{amsfonts}
\usepackage{bbold}

\pdfoutput=1
\newcommand{\ket}[1]{|#1 \rangle}
\newcommand{\bra}[1]{\langle#1 |}

\newcommand{\lr}[1]{\left( #1 \right)}

\newcommand{\mean}[1]{\langle #1 \rangle}
\newcommand{\no}{\nonumber}

\newcommand{\mc}[1]{\mathcal{#1}}

\newcommand{\E}{\mathcal{E}}
\newcommand{\G}{\mathcal{G}}

\newcommand{\Mis}{\textsc{Mis }}

\newcommand{\Toy}{\textrm{Toy}}
\newcommand{\UD}{\textrm{UD}}
\newcommand{\eff}{\textrm{eff}}
\newcommand{\Ryd}{\textnormal{Ryd}}

\newtheorem{theorem}{Theorem}

\begin{document}

\title{Computational complexity of the Rydberg blockade in two dimensions}

\author{Hannes Pichler}
\thanks{These authors contributed equally to this work.}
\affiliation{ITAMP, Harvard-Smithsonian Center for Astrophysics, Cambridge, MA 02138, USA }
\affiliation{Department of Physics, Harvard University, Cambridge, MA 02138, USA }
\author{Sheng-Tao Wang}
\thanks{These authors contributed equally to this work.}
\affiliation{Department of Physics, Harvard University, Cambridge, MA 02138, USA }
\author{Leo Zhou}\thanks{These authors contributed equally to this work.}
\affiliation{Department of Physics, Harvard University, Cambridge, MA 02138, USA }
\author{Soonwon Choi}
\affiliation{Department of Physics, Harvard University, Cambridge, MA 02138, USA }
\affiliation{Department of Physics, University of California Berkeley, Berkeley, CA 94720, USA }
\author{Mikhail D. Lukin}
\affiliation{Department of Physics, Harvard University, Cambridge, MA 02138, USA }

\begin{abstract}
We discuss the computational complexity of finding the ground state of the two-dimensional array of quantum bits that interact via strong van der Waals interactions.  Specifically,  we focus on systems where the interaction strength between two spins depends only on their relative distance $x$ and decays as $1/x^6$ that have been realized with  individually trapped homogeneously excited neutral atoms interacting via the so-called Rydberg blockade mechanism.  
We show that the solution to NP-complete problems can be encoded in ground state of such a many-body system by a proper geometrical arrangement of the atoms.
We present a reduction from the NP-complete maximum independent set problem on planar graphs with maximum degree three. Our results demonstrate that computationally hard optimization problems can be naturally addressed with coherent quantum optimizers accessible in near term experiments.
\end{abstract}

\date{\today}

\maketitle
\section{Introduction}
Various quantum algorithms have been proposed in recent years to solve combinatorially hard optimization problems \cite{Farhi:2001hya,Smelyanskiy:2001vo,Das:2008hz,{Albash:2018he},{Santoro:2002cra},Boixo:2013ha,Johnson:2011gd,{Ronnow:2014fd},{Houck:2012iqa},{Smolin:2013vp},{Wang:2013tg},{Shin:2014wf},{Harrow:ve}}.
The fundamental idea is  based on encoding such problems in the classical ground state of a programmable quantum system, such as spin models \cite{{Albash:2018he}}.
Quantum algorithms are then designed to utilize quantum evolution in order to drive the system into this ground state, such that a subsequent measurement reveals the solution \cite{Farhi:2001hya,Farhi:2014wk,Farhi:2016vm}. 

The models considered in this context are typically based on many-body spin systems \cite{{Albash:2018he}}. Assuming a complete control of the interactions between the spins, it is possible to encode NP-complete optimization problems \cite{Lucas:2014eb} into ground states of such systems. 
In most realizations, however, interactions are not fully programmable, but instead determined by detailed properties of their specific physical realizations.
These properties include locality, geometric connectivity, and controllability, which either constrain the class of problems that can be efficiently realized~\cite{{Albash:2018he}}, or imply that substantial overhead is required for their realization \cite{Lechner:2015iy,Glaetzle:2017hh}. 
Thus, one of the central challenges in understanding and assessing quantum optimization algorithms in near-term devices involves designing  methods to encode  important classes of combinatorial problems in physical systems in an efficient and natural way \cite{Preskill:2018gt}.

\begin{figure}[b]
\begin{center}
\includegraphics[width=0.5\textwidth]{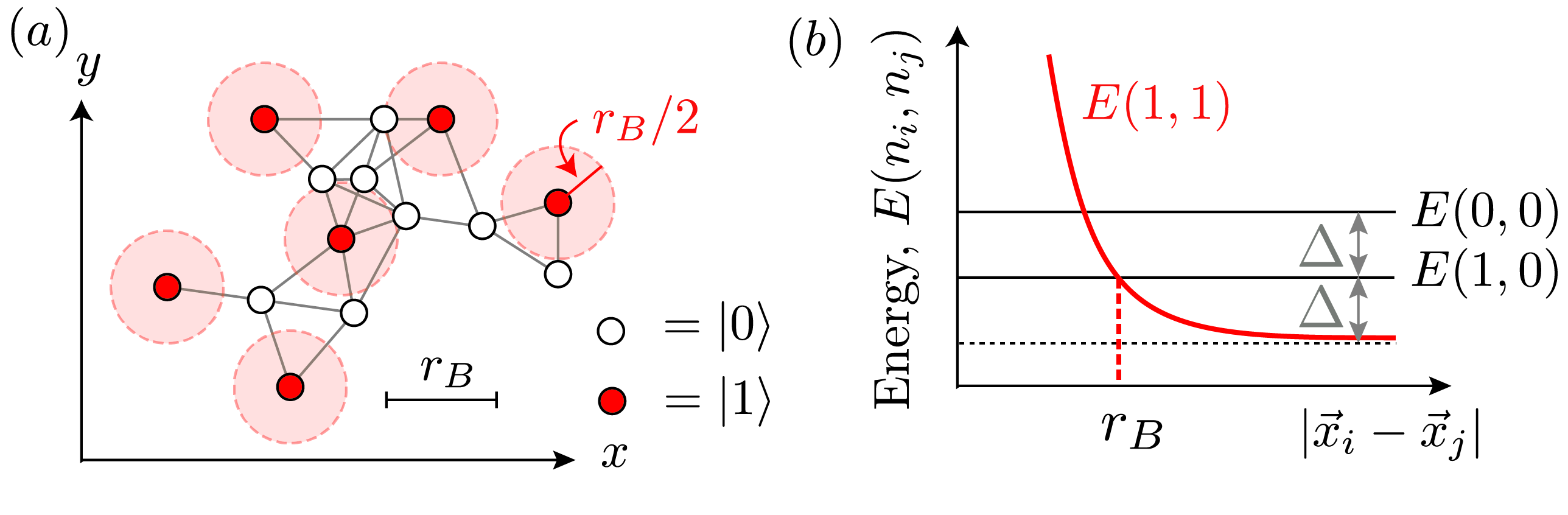}
\caption{(a) Illustration of the Rydberg problem. $N$ atoms are arranged in the 2D plane with positions $\vec{x}_v$. Each atom can be in one of two internal states, $\ket{0}$ (white) or $\ket{1}$ (red). The Rydberg blockade mechanism prevents two atoms from being simultaneously in state $\ket{1}$ if they are within a distance $r_B$ [see (b)]. For clarity, two atoms $i$ and $j$ are connected by a gray line if $|\vec{x}_i-\vec{x}_j|<r_B$. (b) The atoms interact pairwise if in state $\ket{1}$, with an interaction strength that decays as $1/x^6$ where $x$ is the distance  of the two atoms. Energy curves are plotted for a homogeneous detuning $\Delta_v=\Delta>0$ and $\Omega=0$.
The blockade radius $r_B$ corresponds to the distance where the ground state switches from the configuration with both atoms in the Rydberg state, $\ket{1}$, at $|\vec{x}_1-\vec{x}_2|>r_B$, to the configuration where only one of them is in the Rydberg state, at $|\vec{x}_1-\vec{x}_2|<r_B$.
}\label{FigRydberg}
\end{center}
\end{figure}

Recently, significant progress in realizing spins models with trapped neutral atoms has been made \cite{Bernien:2017bp,Labuhn:2016dp,Lienhard:2018fm,GuardadoSanchez:2018jf,Levine:2018uk}. Using optical tweezers, atoms can be individually arranged in fully programmable arrays in one \cite{Endres:2016fka}, two \cite{Barredo:2016ea,Kim:2016hx} and even three dimensions \cite{Barredo:2018ee}. These systems 
use lasers to coherently excite atoms from their internal ground state, $\ket{0}$, to long-lived Rydberg states, $\ket{1}$ \cite{Bernien:2017bp,Labuhn:2016dp,Lienhard:2018fm,GuardadoSanchez:2018jf,Levine:2018uk}. If two atoms are both in this Rydberg state, they interact via strong van der Waals interactions \cite{Saffman:2010ky}. These systems are described by the Hamiltonian
 \begin{align}\label{eq:HRYD}
H_{\rm Ryd}&=\sum_{v}\Omega_v \sigma_v^x-\Delta_v n_v+\sum_{w> v} V_{\rm Ryd}(|\vec{x}_v-\vec{x}_w|)n_v n_w.
\end{align}
Here $\vec{x}_v$ specifies the position of atom $v$. The parameters $\Omega_v$, and $\Delta_v$ characterize the Rabi frequency and the detuning of the coherent laser at the position $\vec{x}_v$. The operator $\sigma_v^x=\ket{0}_v\bra{1}+\ket{1}_v\bra{0}$ gives rise to coherent spin flips of atom $v$ and $n_v=\ket{1}_v\bra{1}$ counts if the atom $v$ is in the Rydberg state.  In this work we consider isotropic Rydberg states, where the interatomic interaction strength depends only on the relative atomic distance, and is given by $V_{\rm Ryd}(|\vec{x}|)=C/|\vec{x}|^6$ \cite{Saffman:2010ky}. 
The strong interactions at short distances energetically prevent two atoms to simultaneously be excited in the Rydberg state $\ket{1}$ if they are within a blockade radius, $r_B$ (see Fig.~\ref{FigRydberg}), resulting in the so-called blockade mechanism \cite{Saffman:2010ky}.

In this Article, we analyze the computational complexity of the problems that can be directly encoded in many-body systems described by $H_\Ryd$.
This is important in light of the potential applications of these systems for quantum optimization \cite{MISPRL,Preskill:2018gt}. While the complete control of the atomic positions results in a high degree of programmability, the geometry of the interaction clearly restricts the class of spin models that can be realized.
Nonetheless, as the main result of this paper, we prove:
\begin{theorem}\label{thm:main}
It is NP-complete to decide whether the ground state energy of $H_\Ryd$ is below a given threshold, when $\Omega_v=0$, as long as the atoms can be positioned arbitrarily in at least two dimensions.
\end{theorem}
\noindent
This key result elaborates on the arguments outlined in our recent work of Ref.~\cite{MISPRL}. Specifically, we present a detailed  analysis that is based on the analogy between the Rydberg Hamiltonian and the maximum independent set problem, \textsc{Mis}. \Mis is a paradigmatic optimization problem, addressing the task of finding the largest disconnected subset of vertices on a given graph $G=(V,E)$ with vertices $V$ and edges $E$ \cite{Garey:1979wr}. Deciding whether a graph has a maximum independent set (MIS) of size larger than a given integer, $a$, is a famous NP-complete problem. 
In general, even approximate optimization is NP-hard  \cite{Garey:1979wr}.
Importantly, \Mis can be formulated as an energy minimization problem, by associating a spin-1/2 with each vertex $v\in V$ and considering the Hamiltonian
\begin{align}\label{Seq:HMIS}
H_{P}=\sum_{v\in V} -\Delta \,n_v+\sum_{\{v,w\}\in E} U n_v n_w.
\end{align}
For $\Delta>0$, $H_P$ favors spins to be in state $\ket{1}$. However, if $U>\Delta$,  it is energetically unfavorable for two spins, $u$ and $v$, to be simultaneously in state $\ket{1}$ if they are connected by an edge ${u,v}\in E$.
Thus, each ground state of $H_P$ represents a configuration where the spins that correspond to vertices in the maximum independent set are in state $\ket{1}$, and all other spins are in state $\ket{0}$.
We refer to such a state as MIS-state, and to $H_P$ as MIS-Hamiltonian. 
The NP-complete decision problem of \Mis becomes deciding whether the ground state energy of $H_P$ is lower than $-a\Delta$.

The MIS-Hamiltonian $H_P$ clearly shares some features with the Rydberg Hamiltonian $H_\Ryd$ in the classical limit of $\Omega_v=0$. The main difference lies in the achievable connectivity of the pairwise interaction, in particular, when arbitrary graphs are allowed in $H_P$. We therefore consider a special, restricted class of graphs, that are most closely related to the Rydberg blockade mechanism. These so-called unit disk (UD) graphs
are constructed when vertices can be assigned coordinates in a plane, and only pairs of vertices that are within a unit distance, $r$, are connected by an edge.
Thus the unit distance $r$ plays an analogous role to the Rydberg blockade radius $r_B$ in  $H_\Ryd$. 
Remarkably, \Mis is NP-complete even when restricted to such unit disk graphs \cite{Clark:1990dz}. The main result of this paper is a generalization of this property to the Rydberg Hamiltonian \eqref{eq:HRYD}: while in contrast to UD graphs, the Rydberg interactions extend beyond the unit distance and are infinitely ranged, we show  that finding the ground state of Rydberg Hamiltonian constitutes an NP-complete problem. 

The central idea of this work is to show that by choosing atom positions in two dimensions and laser parameters, the low energy sector of the Rydberg Hamiltonian $H_\Ryd$ reduces to the (NP-complete) MIS-Hamiltonian $H_P$ on planar graphs with maximum degree 3 \cite{doi:10.1137/0132071}.
This formal reduction is made possible owing to two key ideas.
First, we harness the formation of antiferromagnetic order in the ground state of (quasi) 1D spin chains at positive detuning, due to the Rydberg blockade mechanism. This allows us to effectively transport the blockade constraint between distant spins. 
Second, we introduce a detuning pattern, $\{\Delta_v\}$, that eliminates the effect of undesired long-range interactions without altering the ground state spin configurations.  
Based on this reduction, we provide a constructive prescription to efficiently encode NP-complete problems in the ground state of arrays of trapped neutral atoms \footnote{We note our construction constitutes a formal proof of Theorem~\ref{thm:main}, but should not necessarily be understood as a recipe for experimental implementation.}. 
This leads to the proof of Theorem~\ref{thm:main}, where we show finding the ground state energy of interacting Rydberg atoms in 2D array is NP-hard (and NP-complete when $\Omega_v=0$).

Our results imply that quantum optimization algorithms can be tested on computationally hard problems with coherent quantum optimizers accessible in near term experiments  with minimal resources and experimental overhead \cite{MISPRL}.

Our manuscript is organized as follows. 
In Sec.~\ref{Sec:MIS} we review the maximum independent set problem on unit disk graphs. Sec.~\ref{Sec:ToyModel} introduces a simple toy model with finite range interactions that can be interpreted as an approximation to the full Rydberg Hamiltonian. This illustrates the main ideas employed in the discussion of the full Rydberg Hamiltonian in Sec.~\ref{Sec:NPRYD}, but avoids many of the technical subtleties associated with infinitely ranged interactions. In Sec.~\ref{Sec:NPRYD} we address the latter problem and show that the Rydberg problem is NP-complete. 

\section{Maximum independent sets for unit disk graphs}\label{Sec:MIS}
\label{Sec:NPCompletenessUD}

In this section we discuss \Mis problem on unit disk graphs.
The problem can be formulated as minimizing the energy of
\begin{align}
\label{Seq:UDHam}
\begin{split}
H_{\rm UD} &= \sum_{v\in V} -\Delta n_v+\sum_{w> v} V_{\rm UD}(|\vec{x}_v-\vec{x}_w|)n_v n_w,\\
\text{where} &\quad V_{\rm UD}(x) = \begin{cases}
U, & x\le r \\
0, & x > r
\end{cases}
\end{split}
\end{align}
which is a subclass of $H_P$ \eqref{Seq:HMIS}.
Clark et al.\ \cite{Clark:1990dz} has proven NP-completeness of this problem, by reducing it from \Mis on planar graphs of maximum degree 3. 
Since our analysis of the computational complexity associated with Rydberg Hamiltonian \eqref{eq:HRYD} in Sec.~\ref{Sec:NPRYD} is based on a similar reduction, we find it instructive to review the following theorem and its proof.

\begin{theorem}[Clark, Colbourn, and Johnson, 1990 \cite{Clark:1990dz}]
\label{thm:clark}
\Mis on unit disk graphs is NP-complete. 
\end{theorem}

\begin{proof}
(1) \Mis on planar graphs with maximum degree 3 is NP-complete.
(2) A planar graph $\mc{G}=(\mc{V},\mc{E})$, with vertices $\mc{V}$ and edges $\mc{E}$, with maximum degree 3 can be embedded in the plane using $\mc{O}(|\mc{V}|)$ area in such a way that its vertices are at integer coordinates, and its edges are drawn by joining line segments on the grid lines $x=i\times g$ or $y=j\times g$, for integers $i$ and $j$, and grid spacing $g$.
(3) One can replace each edge $\{u, v\}$ of $\mc{G}$ by a path having an even number, $2k_{u,v}$, of ancillary vertices, in such a way that a UD graph, $G=(V,E)$ with vertices $V$ and edges $E$, can be constructed. For an explicit construction, see below. It is straightforward to verify that $\mc{G}$ has an independent set $\mc{S}\subset\mc{V}$ such that $|\mc{S}|\leq a$ if and only if $G$ has an independent set $S\subset V$ such that $|S|\leq a' \equiv a+\sum_{\{u,v\}\in \mc{E}} k_{u,v}$.
\end{proof}

The above theorem shows that it is NP-complete to decide whether the ground state energy of $H_{\rm UD}$ is lower than $-a'\Delta$.

\begin{figure}[t]
\begin{center}
\includegraphics[width=0.5\textwidth]{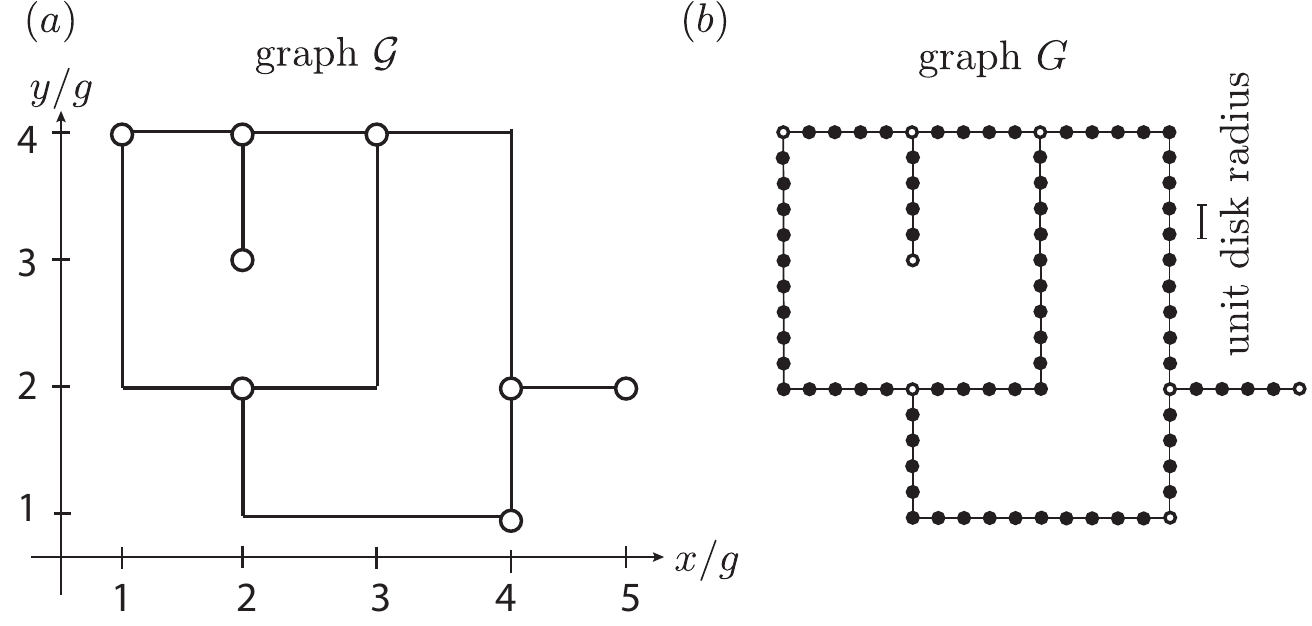}
\caption{Grid representation of a planar graph of maximum degree 3 (a) and a transformation to a unit disk graph (b). The ancillary vertices are indicated as solid black circles.}\label{FigSimple}
\end{center}
\end{figure}

\subsubsection{Vertex arrangement  prescription}\label{Sec:VertexArrangement}  
The transformation in the proof of Theorem~\ref{thm:clark} does not fully determine the actual positions of the ancillary vertices in the 2D plane. We find it convenient to specify a particular arrangement consistent with the requirements of this transformation. This simplifies the discussion, in particular once we consider Rydberg interactions, as it fixes the interaction strength between each pair of spins.
 
Consider an edge $\{u,v\}$ of the graph $\mc{G}$ embedded on the grid. Let us denote the length of this edge by $g\times \ell_{u,v}$ with $\ell_{u,v}$ integer, and $g$ the grid unit length. First, place an ancillary vertex on the $\ell_{u,v}-1$ grid point along the edge, separating the edge in $\ell_{u,v}$ segments of length $g$. We choose an integer $k\geq3$ and place equally spaced $2k$ ancillary vertices along each segment, dividing it into $2k+1$ pieces of equal length $d=g/(2k+1)$. 
If $\ell_{u,v}$ is even, we choose one segment and replace the $4\phi+2$ ancillary vertices close to the center of this segment with $4\phi+1$ ancillary vertices, that are equally spaced by a distance $D=d+\frac{d}{4\phi}$, for some integer $\phi$ to be determined.
We refer to such segments as ``irregular segments'', and to the vertex at the center of the irregular segment as irregular vertex. These exceptions are made to ensure that the total number of ancillary vertices along each edge $\{u,v\}\in\E$, $2k_{u,v}$, is even.
Following this arrangement, the nearest-neighbor distance of the ancillary vertices is either $d$ or $D$.
Setting the unit disk radius to $r=D+0^+$ produces the unit disk graph $G$. The positions of the vertices are labelled by $\vec{x}_v$. Note that this arrangement depends on the freely chosen parameters $k$ and $\phi$.\\

\subsubsection{Structure of the maximum independent set}\label{MIS_Structure}
To prepare for the discussion in the following sections, we note a few properties of the maximum independent set on the types of unit disk graphs constructed in this way. First, note that the maximum independent set on $G$ is in general degenerate, even if the maximum independent set on $\G$ is unique. We are in the following interested in the properties of a particular set of MIS-states on $G$. Specifically, we consider such MIS-states on $G$, $\ket{\psi_G}$, that coincide with MIS-states on $\G$, $\ket{\psi_\G}$, on the vertices $\mc{V}$. To see that such states always exist, we can simply construct $\ket{\psi_{G}}$ from $\ket{\psi_{\G}}$ as follows. 
Consider the state of two spins corresponding to two vertices, $u,v\in \mc{V}$. If $u$ is in the state $\ket{0}$ and $v$ is in state $\ket{1}$ (or vice versa) then we place the $2k_{u,v}$ ancillary vertices on the edge connecting $u$ and $v$ in an (antiferromagnetically) ordered state, where the spins are alternating in states $\ket{1}$ and $\ket{0}$. 
In the other case, if $u$ and $v$ are both in states $\ket{0}$, we place the ancillary vertices in an analogously ordered state, however in this case we need to introduce a domain wall, that is, an instance where two neighboring spins are both in the state $\ket{0}$. The position of this domain wall along the edge is clearly irrelevant. 
In both cases half of the $2k_{u,v}$ ancillary vertices along the edge are in state $\ket{1}$. Therefore, by the above theorem, the state constructed from $\ket{\psi_\G}$ by applying this process on all edges of $\G$, is a MIS-state on $G$.

Below we will be particularly interested in the structure of these MIS-states around points where edges meet under a 90$^\circ$ angle. Such points are either junctions, where 3 edges meet at a vertex, or corners (see Fig.~\ref{FigSimple}). Importantly we note, that there is a MIS-state, $\ket{\psi_G}$, such that the spins close to each corner or junction are ordered. More precisely, there exists a maximum independent set, such that for every corner and junction, all vertices within a distance $\delta<g/4$ are in one of the two possible ordered configurations with no domain wall.
This simply follows from the above discussion by noting that the position of any domain wall along an edge $\mc{E}$ is irrelevant. In particular any domain wall can be moved along an edge $\mc{E}$ such that its distance to any vertex on a grid point is larger than $g/4$. The possibility to move domain walls is also exploited in the discussion of the full Rydberg problem in Sec.~\ref{Sec:NPRYD}.

\section{Toy model}\label{Sec:ToyModel}

In order to illustrate the conceptual ideas in generalizing the above reduction to the case of Rydberg interactions, we first consider a simple toy model. In particular, we consider a Hamiltonian similar to $H_\UD$, the MIS-Hamiltonian for UD graphs, but introduce interactions beyond the unit disk radius. Similar to the situation in the Rydberg system, these additional interactions cause energy shifts that can result in a change of the ground state, thus invalidating the encoding of the $\textsc{Mis}$. The key idea to resolve this issue is to use local detunings that compensate for the additional interactions.

Specifically, we consider the model given by
\begin{align}\label{Seq:HTOY}
H_{\rm Toy}=\sum_{v\in V} -\Delta_v \,n_v+\sum_{w> v} V_{\rm{Toy}}(|\vec{x}_w-\vec{x}_v|) n_v n_w.
\end{align}
with interactions 
\begin{align}
V_{\rm Toy}(x)=\left\{\begin{array}{cl}U, & x\leq r,\\ W, & r<x\leq R \\0, & x>R, \end{array}\right.
\end{align}
where $W<U$.  
Clearly for $W=0$ and $\Delta_v=\Delta$, $H_\Toy$ reduces to the Hamiltonian $H_{\rm UD}$  \eqref{Seq:UDHam}.
For $W>0$ it includes interactions beyond the unit disk radius, $r$, and can thus be considered as a first approximation to the Rydberg Hamiltonian \cite{Samajdar:2018cg}. 

Let us consider a spin arrangement as in Sec.~\ref{Sec:VertexArrangement} corresponding to a unit disk graph, $G$, and the case $\sqrt{2}r<R<2r$. In this case most spins interact only with their neighbors on $G$, with the only exception being spins that are close to corners or junctions. There, due to the geometric arrangement, spins are interacting (with strength $W>0$) even though they are not neighboring on $G$. These interactions are relevant and can potentially change the ground state, as compared to the MIS-Hamiltonian, $H_P$. For a simple example of such a situation see Fig.~\ref{FigExampleTOY}(a).

\begin{figure}[t]
\begin{center}
\includegraphics[width=0.5\textwidth]{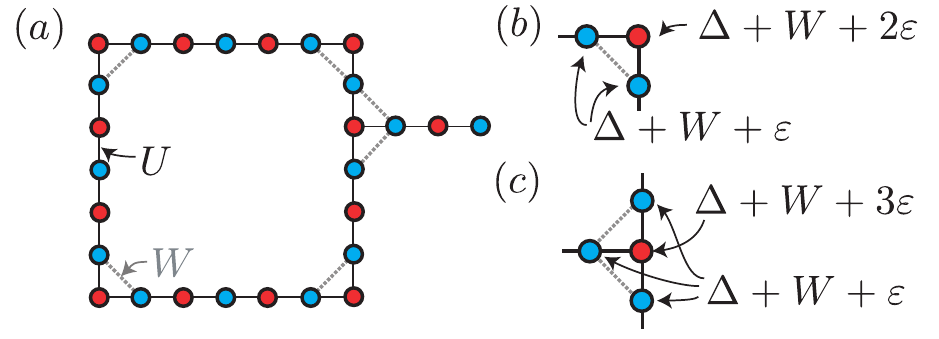}
\caption{(a) Example of a vertex arrangement corresponding to a unit disk graph, where the ground state of the toy model $H_{\Toy}$ \eqref{Seq:HTOY} does not correspond to the maximum independent set. The maximum independent set on this graph is indicated by the blue vertices. The red vertices indicate another independent set whose size is smaller than the size of the maximum independent set by one vertex. Clearly, if $W$ is too large, the additional interactions on the diagonals close to corners and junctions disfavor the MIS-state; instead, the state, where the spins on the red sites are in state $\ket{1}$, becomes the ground state. This problem can be resolved by changing the detuning locally at the problematic structures as indicated in (b) and (c) (see main text for discussion).}\label{FigExampleTOY}
\end{center}
\end{figure}

As mentioned above, we want to find a detuning patten, $\Delta_v$, such that the MIS-state of $H_\UD$ remains the ground state of $H_{\rm Toy}$, even at finite $W$. 
We first note that for the relevant spin arrangements, the interactions of the $H_{\UD}$ and  $H_\Toy$ differ only around corners and junctions. Thus we set $\Delta_v=\Delta$ everywhere, except at these structures and consider them individually and separately.

First, let us consider a corner vertex, as in Fig.~\ref{FigExampleTOY}(b). As discussed in Sec.~\ref{MIS_Structure}, there exists at least one MIS-state, $\ket{\psi_G}$, such that the corner spin and its two neighbors are in one of the two ordered configurations: either the corner spin is in state $\ket{0}$ and its two neighbors are in state $\ket{1}$ ($\ket{011}$), or the corner vertex is in state $\ket{1}$ and its two neighbors are in state $\ket{0}$ ($\ket{100}$). The idea is to choose detunings
for these three spins in such a way that only these two configurations are relevant for the ground state of $H_\Toy$. That is, we want to ensure that the energy of all possible spin configurations can always be lowered by arranging the spins on the corners in one of the two ordered states. This can be achieved if we choose the detuning of the corner vertex (labeled by $C$) to be $\Delta_C=\Delta+W+2\varepsilon$, and the detuning of its two neighbors to be $\Delta_N=\Delta+W+\varepsilon$. Here $\varepsilon>0$ can be chosen freely (up to the trivial constraint $\Delta_v<U$).
Importantly, this choice of the detuning restores the energy difference between the two ordered configurations on the three spins to $\Delta$, corresponding to the additional vertex in state $\ket{1}$.
Hence, up to a trivial constant (and contributions from junctions discussed below), $H_{\rm Toy}$ and $H_{\rm UD}$ are identical on all states where each corner is in one of the above ordered configurations.
This includes in particular one ground state of $H_{\rm UD}$, i.e. a MIS-state. All states that are not of this type have higher energy with respect to $H_{\rm Toy}$ by construction.

Junctions can be treated analogously [see Fig.~\ref{FigExampleTOY}(c)]. Again, there exists at least one MIS such that the central vertex is in the state $\ket{0}$ and its three neighbors are in the state $\ket{1}$, or the central vertex is in the state $\ket{1}$ and its three neighbors are in the state $\ket{0}$. We choose the detuning of the degree 3 vertex (denoted  $\Delta_J$) to be $\Delta_{J}=\Delta+W+3\varepsilon$, and the detuning of its three neighbors to be $\Delta_{N'}=\Delta+W+\varepsilon$. Again, this choice renders one of the two ordered states on the junction energetically more favorable than any other state, and restores their energy difference to exactly $2\Delta$.

In summary, the actions of $H_{\rm UD}$ and $H_{\rm Toy}$ are identical for at least one MIS-state, for the above choice of the detuning pattern. In addition, we ensure that our choice of detunings cannot lower the energy of any other configurations. Therefore, a ground state of $H_{\rm Toy}$ is a ground state of $H_\UD$, encoding an MIS configuration on the corresponding unit disk graph.

\section{NP-comleteness of the Rydberg problem}\label{Sec:NPRYD}
We now generalize the arguments presented above for the toy model to the case of the Rydberg Hamiltonian, allowing us to prove Theorem~\ref{thm:main}.
While the main idea is similar, the infinite ranged Rydberg interactions require a more careful discussion.

We first show in Sec.~\ref{Sec:separation_of_scales} that one can separate length and interaction scales. This is possible since the Rydberg interactions decay sufficiently fast, such that we can separate the interactions between spins that are close, from the interactions between spins that are far apart on the graph $G$. In particular, we show that the interactions between distant spins can be neglected, if $k$ is chosen large enough. This allows us to treat individual substructures of the system separately. 

We then argue in Sec.~\ref{Sec:effectiveModel} that one can map the low energy sector of the Rydberg Hamiltonian to a much simpler effective spin model. To this end we consider clusters of spins and choose specific detuning patterns, such that only two configurations are relevant for each cluster. We show that the resulting, effective pseudo-spins are described by a MIS Hamiltonian. This allows us to encode \Mis on planar graphs with maximum degree 3 to the ground state of the Rydberg Hamiltonian, proving Theorem~\ref{thm:main}.
The remaining Sec.~\ref{Sec:Low_energy}--\ref{Sec:Effective_energy} are dedicated to proving the details of the mapping to the effective model.

We emphasize that the discussion in this section is aimed to obtain a formal proof Theorem~\ref{thm:main}. It should not necessarily be understood as a recipe for experimental implementation of quantum optimization algorithms for \Mis discussed in Ref.~\cite{MISPRL}. 

\subsection{Separation of length scales}\label{Sec:separation_of_scales}
The Rydberg interactions decay as a power law with distance and thus do not define a length scale. Nevertheless, the above vertex arrangement introduces two length scales, the closest distance between two spins, $d$, and the grid length, $g$. They are separated by a factor, $g=d\,(2k+1)$, that can be chosen arbitrarily in the transformation of the planar graph, $\G$, to the unit disk graph, $G$. This allows us to separate interactions between spins that are ``close'' from interactions between spins that are ``distant''. 

The closest spins in the system are separated by a distance $d$ and thus interact with a strength of $\mc{U}=C/d^6$. This defines a convenient unit of energy. Note that $\mc{U}$ depends on the choice of the parameter $k$ specifying the transformation from $\G$ to $G$.

For concreteness, we define two spins to be ``distant'' if their $x$ or $y$ coordinate differs by at least $g$ \footnote{This separation is of course arbitrary. Any finite value in units of $g$, that is, independent of $k$, works as well.}. 
It is easy to see that the interaction energy of a single spin with all distant spins is upper bounded by 
\begin{align}
E_{\rm dist}&=C\left(\sum_{j=1}^{\infty} \sum_{i=1}^{\infty}\frac{8}{(d^2i^2+g^2j^2)^3}+\sum_{i=0}^{\infty} \frac{4}{(id+g)^6}\right).\no
\end{align}
This bound can be obtained by considering a system where we place a spin in state $\ket{1}$ on each possible spin position in the 2D plane (i.e. along all grid lines).  $E_{\rm dist}$ is simply the interaction energy of a single spin in such a system with all other spins that are distant in the above sense. This is clearly an upper bound to the maximum interaction energy of a single spin with all distant spins on arbitrary graphs, that itself can be upper bounded straightforwardly by
\begin{align}\label{Seq:distantVert}
E_{\rm dist}&\leq \mc{U}\left(\frac{3 \pi  \zeta (5)}{2}+\frac{4}{5}+4\lr{\frac{d}{g}}\right)\lr{\frac{d}{g}}^5,
\end{align}
where $\zeta(x)$ is the Riemann zeta function.
Since $d/g=1/(2k+1)$, we find that  the interaction energy of a spins with all distant spins scales as $E_{\rm dist}=\mc{O}(\mc{U}/k^5)$. Note that the total number of spins scales as $|V|\sim \mc{O}(k|\mc{V}|)$, as the grid representation of the original planar graph $\G$ in the plane can be done with $\mc{O}(|\mc{V}|)$ area.
Thus the contribution of interactions between all pairs of spins that are distant with respect to each other is bounded by $E_{\rm dist}|V|\sim \mc{O}(\mc{U}|\mc{V}|/k^4)$.
We can therefore always choose $k$ large enough such that these contributions can be neglected. Note that the required $k$ scales only polynomially with the problem size $|\mc{V}|$.

\subsection{Effective spin model}\label{Sec:effectiveModel} 
\begin{figure*}[t]
\begin{center}
\includegraphics[width=1\textwidth]{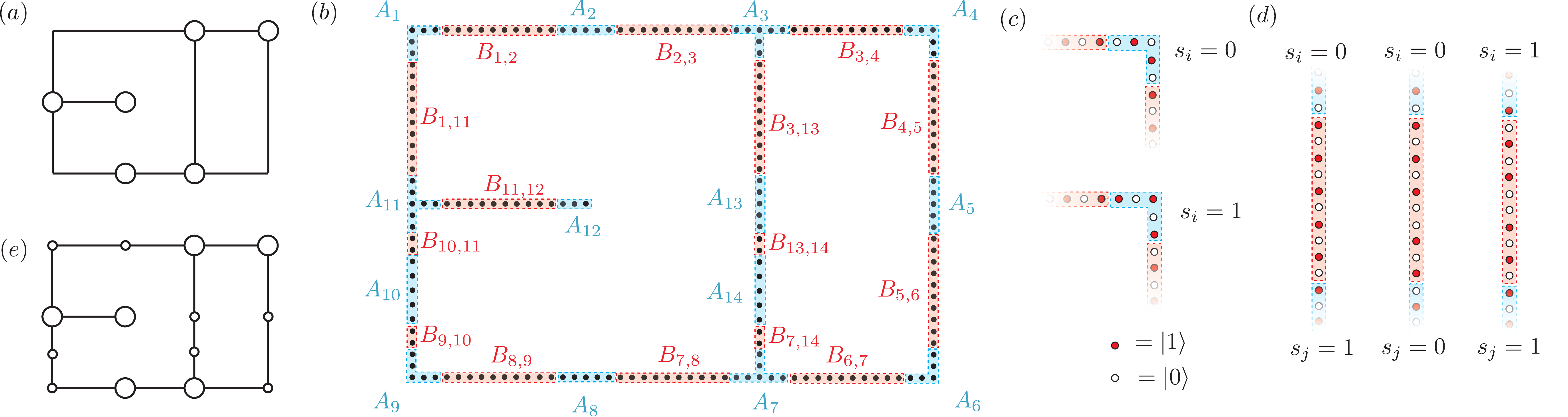}
\caption{Transformation to the effective spin model. (a) Example of a planar graph with maximum degree 3, $\G$ embedded on a grid. (b) Corresponding spin arrangement following the prescription given in Sec.~\ref{Sec:VertexArrangement}, with $k=7$ and $\phi=1$. The regions $A_i$, corresponding to a special vertex and its $2q$ neighbors on each leg are highlighted in blue (here $q=1$). We call regions of the type $A_1$, $A_4$, $A_6$ and $A_9$ corners, $A_3$, $A_7$, and $A_{11}$ junctions, $A_{12}$ open legs, $A_2$, $A_{5}$, $A_8$, $A_{10}$, $A_{13}$ and $A_{14}$ as special vertices on straight segments. Note that $A_{10}$ and $A_{14}$ are irregular regions. In the effective spin model, each region $A_i$ is represented by a pseudo-spin $s_i$ that interacts with neighboring pseudo-spins. The two pseudo-spin states $s_i=0,1$ correspond to the two ordered configurations in $A_i$ (c). The state of the segments connecting two pseudo-spins $i$ and $j$ are determined by the pseudo-spin states $s_i$ and $s_j$ (up to the location of the domain wall) (d). The graph $\G^{\rm eff}$ corresponding to the effective spin model is depicted in (e). Note that this $\G^{\rm eff}$ is obtained from $\G$ by replacing each edge by a string of even ancillary vertices (small circles). Thus the MIS on $\mc{G}$ differs from the MIS on $\G^{\rm eff}$ by half the number of such ancillary vertices.  }\label{Fig_separation}
\end{center}
\end{figure*}

Similar to the example of the toy model in Sec.~\ref{Sec:ToyModel}, interactions beyond the unit disk radius are problematic, and they are most prominent close to vertices of degree 3 or corners. To address these complications, we make use of the separation of length and interaction scales (Sec.~\ref{Sec:separation_of_scales}) to isolate these structures and analyze them individually.
For the graphs discussed in this work whose vertex arrangement is described in Sec.~\ref{Sec:VertexArrangement}, we define various  ``special'' vertices: these are vertices that either have integer grid coordinates or are irregular vertices. 
Special vertices thus include all vertices $\mc{V}$ of $\mc{G}$, but also a subset of vertices introduced to transform $\G$ to the UD graph $G$. 
For each special vertex, $i$, we define a set of vertices, $A_i$, that includes $i$ and its $2q$ neighbors on each leg (see Fig.~\ref{Fig_separation}). Here we chose $q=\lfloor k/8\rfloor$ \footnote{Any choice of $q$ works, as long as we choose $2\phi<2q\leq k/4$ and $q\sim \mc{O}(k)$}. We further define sets $B_{i,j}$ consisting of the vertices along the path connecting $A_i$ and $A_j$. Note that by construction $B_{i,j}$ contains an even number of vertices. 

Due to the separation of length and interaction scales (Sec.~\ref{Sec:separation_of_scales}), interactions between two spins, $u$ and $v$, are only relevant if they belong to the same set, $u,v\in A_i$ (or $u,v\in B_{i,j}$), or if they belong to adjacent sets, $u\in A_i$ and $v\in B_{i,j}$. All other interaction terms will give contributions that vanish like $\mc{O}(\mc{U}|\mc{V}|/k^4)$. The total energy can thus be written as 
\begin{align}
E=\!\sum_{i}E_{A_i}+\sum_{\mean{\mean{i,j}}}\lr{E_{A_i,B_{i,j}}+\frac{1}{2}E_{B_{i,j}}}\!+\!\mc{O}(\mc{U}|\mc{V}|/k^4).
\end{align}
Here, $E_{X}=\sum_{v\in X}-\Delta_v n_v+\sum_{w> v\in X}V_{\rm Ryd}(|\vec{x}_w-\vec{x}_v|)n_w n_v $ gives the energy of the system when only spins in $X$ are taken into account, and 
$E_{X,Y}=\sum_{u\in X}\sum_{v\in Y}V_{\rm Ryd}(|\vec{x}_u-\vec{x}_v|)n_un_v$
quantifies the interaction energy between spins in two regions $X$ and $Y$. 
The sum $\sum_{\mean{\mean{i,j}}}=\sum_i\sum_{j\in\mc{N}(i)}$ runs over $i$ and $j$, such that $A_i$ is connected to $A_j$, by a non-empty set $B_{i,j}$, and the factor $1/2$ compensates for double counting.   

Below, in Secs.~\ref{Sec:Low_energy} and \ref{Sec:reduction}, we will show that, with an appropriate choice of the detunings, $\Delta_v$, only a few configurations of spins in regions $A_i$ and $B_{i,j}$ are relevant to construct the ground state of the entire system. This will be a crucial point in our discussion.

In particular, we show in Sec.~\ref{Sec:reduction} that only two configurations of spins in the regions $A_i$ are relevant [see also Fig.~\ref{Fig_separation}(c)]. These two configurations, denoted by $s_i=0$ and $s_i=1$, correspond to states where the special spin, $i$, is in state $\ket{s_i}$ and all other spins in $A_i$ are in the corresponding perfectly ordered state. For each region $A_i$ we can thus define a pseudo-spin $s_i$, and calculate $E_{A_i}$ for both pseudo-spin states (denoted $E_{A_i}^{s_i}$). Note that the number of spins in $A_i$ that are in state $\ket{1}$ is directly determined by the pseudo-spin state $s_i$, as it is given  by $s_i+m_i q$ (where $m$ denotes the degree of the special vertex $i$).

Second, we show in Sec.~\ref{Sec:Low_energy} that only four configurations of spins in each region $B_{i,j}$ are relevant [see also Fig.~\ref{Fig_separation}(d)]. Moreover, these configurations are completely determined by the pseudo-spin state $s_i$ and $s_j$. If $s_i=1$ and $s_j=0$ (or vice versa) the total energy is minimized if $B_{i,j}$ is in the perfectly ordered state. If $s_i=s_j=0$ or $s_i=s_j=1$, the lowest energy state is also ordered but with a single domain wall in the center of the region $B_{i,j}$. We denote the energy $E_{B_{i,j}}$ for these configurations by $E_{B_{i,j}}^{s_i,s_j}$, and therefore write
\begin{align}
E_{B_{i,j}}&=(s_{i}(1-s_j)+(1-s_i)s_j)E_{B_{i,j}}^{1,0}\no\\&+(1-s_i)(1-s_j)E_{B_{i,j}}^{0,0}+s_is_jE_{B_{i,j}}^{1,1},
\end{align}
where we note $E_{B_{i,j}}^{1,0}=E_{B_{i,j}}^{0,1}$.
Note that the number of spins in state $\ket{1}$ in $B_{i,j}$ is given by $b_{i,j}-s_is_j$ (where $2b_{i,j}$ gives the number of spins in $B_{i,j}$).
Moreover, it is easy to see that the interaction energy $E_{A_i,B_{i,j}}$ is constant for all of the above combinations of relevant states [up to $\mc{O}(\mc{U}/k^5)$]. 

The total energy in the relevant configuration sector containing the ground state of $H_{\rm Ryd}$
is therefore given by an effective spin model for the pseudo-spins [analogous to \eqref{Seq:HMIS} and \eqref{Seq:UDHam}]
\begin{align}\label{eq:effectiveMIS}
E=&\sum_{i}-\Delta^{\rm eff}_is_i+\sum_{{\mean {i,j}}}s_{i}s_jU_{i,j}^{\rm eff} + \xi + O(\mc{U}|\mc{V}|/k^4),
\end{align}
where the sum $\sum_{\mean{i,j}}$ runs over neighboring pairs pseudo-spins (without double-counting), and we introduced effective detunings $\Delta^{\rm eff}_i$ and pseudo-spin interactions $U^{\rm eff}_{i,j}$ given by
\begin{align}
\Delta^{\rm eff}_{i}&=E_{A_i}^{0}-E_{A_i}^1-\sum_{j\in\mc{N}(i)}(E_{B_{i,j}}^{1,0}-E_{B_{i,j}}^{0,0}),\label{Def_eff_detuning}\\
U_{i,j}^{\rm eff}&=E_{B_{i,j}}^{1,1}+E_{B_{i,j}}^{0,0}-2E_{B_{i,j}}^{1,0},\label{Def_eff_interactions}\\
\xi &= \sum_i E_{A_i}^0 + \sum_{\mean{\mean{i,j}}} ( E_{A_i,B_{i,j}} + \frac12 E_{B_{i,j}}^{0,0}).
\end{align}
Importantly, we can choose detuning patterns, $\Delta_v$, such that the effective pseudo-spin detunings and interactions are homogeneous, $\Delta_{i}^{\rm eff}=\Delta^{\rm eff}$ and $U_{i,j}^{\rm eff}=U^{\rm eff}$, and satisfy $0<\Delta^{\rm eff}<U^{\rm eff}$. 
To ensure the long-range correction is bounded by some small constant $\eta \ll \Delta^\eff$, we only need to choose $k\ge \mc{O}((\mc{U}|\mc{V}|/\eta)^{1/4})$.
We can also efficiently compute $\xi$ based on the chosen detuning pattern.

Finally, we can relate this effective spin model \eqref{eq:effectiveMIS} back to the original MIS problem on $\G$.
To this end, note that the effective spin model corresponds exactly to a MIS problem on a graph $\mc{G}^\eff$ obtained from $\G=(\mc{V},\mc{E})$ by replacing each edge $\{u,v\}\in\mc{E}$ by a string of an even number $2\kappa_{u,v}$ of vertices (see Fig.~\ref{Fig_separation}). These additional vertices transport the independent set constraint giving an exact correspondence between the maximum independent sets of $\mc{G}$ and $\mc{G}^\eff$, in complete analogy to the discussion in Sec.~\ref{Sec:NPCompletenessUD}.

In summary, we conclude that for each planar graph $\mc{G}=(\mc{V},\mc{E})$ of maximum degree 3 one can thus efficiently find an arrangement of $\mc{O}(k|\mc{V}|)=\mc{O}(|\mc{V}|^{5/4})$ atoms, and detuning patterns such that the ground state of the corresponding Rydberg Hamiltonian directly reveals a maximum independent set of $\mc{G}$.
This allows us to prove our main result of Theorem~\ref{thm:main}:

\begin{proof}[Proof of Theorem~\ref{thm:main}]
To show that it is NP-complete to decide whether the ground state energy of $H_\Ryd(\Omega_v=0)$ is lower than some threshold, we first note that it is in NP.
This is trivial since the ground state of $H_\Ryd(\Omega_v=0)$ has a classical description and serves as a proof if the answer is yes.
Now we show reduction from the NP-complete problem of deciding whether the size of MIS on $\mc{G}$, a planar graph of maximum degree 3, is $\ge a$.
Let $a'=a+\sum_{\{u,v\}\in\mc{E}} \kappa_{u,v}$.
On one hand, if the size of MIS of $\mc{G}$ is $\ge a$, then the ground state energy of $H_{\rm Ryd}$ is $\le -a'\Delta^\eff+\xi+\eta$.
On the other hand, if the size of MIS of $\mc{G}$ is $\le a-1$, then  the ground state energy of $H_{\rm Ryd}$ is $\ge -(a'-1)\Delta^\eff + \xi - \eta$.
By choosing $\eta < \Delta^\eff/2$, the MIS of $\mc{G}$ has size $\ge a$ if and only if the ground state energy of $H_\Ryd$ is $\le -(a'-1/2) \Delta^\eff + \xi$. 
\end{proof}

The remaining part of this paper is devoted to elaborating on the various assumptions that underlie this mapping from the original Rydberg Hamiltonian \eqref{eq:HRYD} to the effective MIS problem \eqref{eq:effectiveMIS}.

\subsection{Low energy sector}\label{Sec:Low_energy}
We find it useful to make a few general remarks about the structure of the ground state of the Rydberg Hamiltonian, and note a few simple properties. 

\subsubsection{Maximal independent sets}\label{Sec:MaximalIS}
In this subsection, we show that the ground state of $H_\Ryd$ corresponds to a maximal independent set on the associated UD graph if the detuning of each spin is in a proper range. Maximal independent sets are independent sets where no vertex can be added without violating the independence property. Clearly the largest maximal independent set is the maximum independent set. For configurations corresponding to such sets, clearly no two neighboring spins can be in state $\ket{1}$ (independence), and no spin can be in state $\ket{0}$ if all its neighbors are in state $\ket{0}$ (maximality). 

For the Rydberg Hamiltonian it is easy to see that no two neighboring vertices on $G$ can be simultaneously in state $\ket{1}$ if the system is in the ground state of $H_{\rm Ryd}$ as long as the detuning $\Delta_v$ on all vertices obeys
\begin{align}\label{max_detuning}
\Delta_v<\Delta_{\rm max}\equiv C/D^6\xrightarrow{\phi\gg1} \mc{U},
\end{align} where $D$ is the maximal Euclidean distance between two vertices that are neighboring on $G$. 

While for the MIS Hamiltonian $\Delta>0$ is sufficient to energetically favor spins to be in the state $\ket{1}$ and thus maximality of the ground state, this is no longer the case for the Rydberg Hamiltonian. In this case, the energy gain $\Delta_v$ has to exceed the energy cost associated with the interaction energy of a spin $v$, with all other spins that are in state $\ket{1}$. In order to determine the required detuning that guarantees maximality, we bound the interaction energy of a spin, $v$ whose neighbors (on $G$) are all in state $\ket{0}$, with the rest of the system. In order to obtain a bound, we split the energy to a contribution from distant spins (bounded by $E_{\rm dist}=\mc{O}(\mc{U}/k^5)$, see \eqref{Seq:distantVert}) and a contribution from the remaining ``close'' spins, $B_1$. The latter is maximal if the spin $v$ is directly neighboring to a vertex of degree 3. Since no two neighboring spins can be simultaneously in the state $\ket{1}$ as long as $\Delta_v<\Delta_{\rm max}$, it can be upper bounded by $B_1\leq \mc{U}\sum_{i=1}^{\infty}\frac{1}{(2i)^6}+2\mc{U}\sum_{i=1}^{\infty}\frac{1}{((2i-1)^2+1)^3}=0.268031\times \mc{U}$.
The first term corresponds to the maximum interaction of spin $v$ with spins on the same grid line, while the second term bounds the interaction with spins on the two perpendicular  gridlines [see also Fig.~\ref{Fig_maximal_sets}(a)].

In summary, we find that the configuration of spins in the ground state of the Rydberg Hamiltonian corresponds to a maximal independent set on the associated UD graph (edge between nearest neighbor) if, $C/D^6\geq\Delta_v \geq 0.268031\times C/d^6$. 

\begin{figure}[t]
\begin{center}
\includegraphics[width=0.45\textwidth]{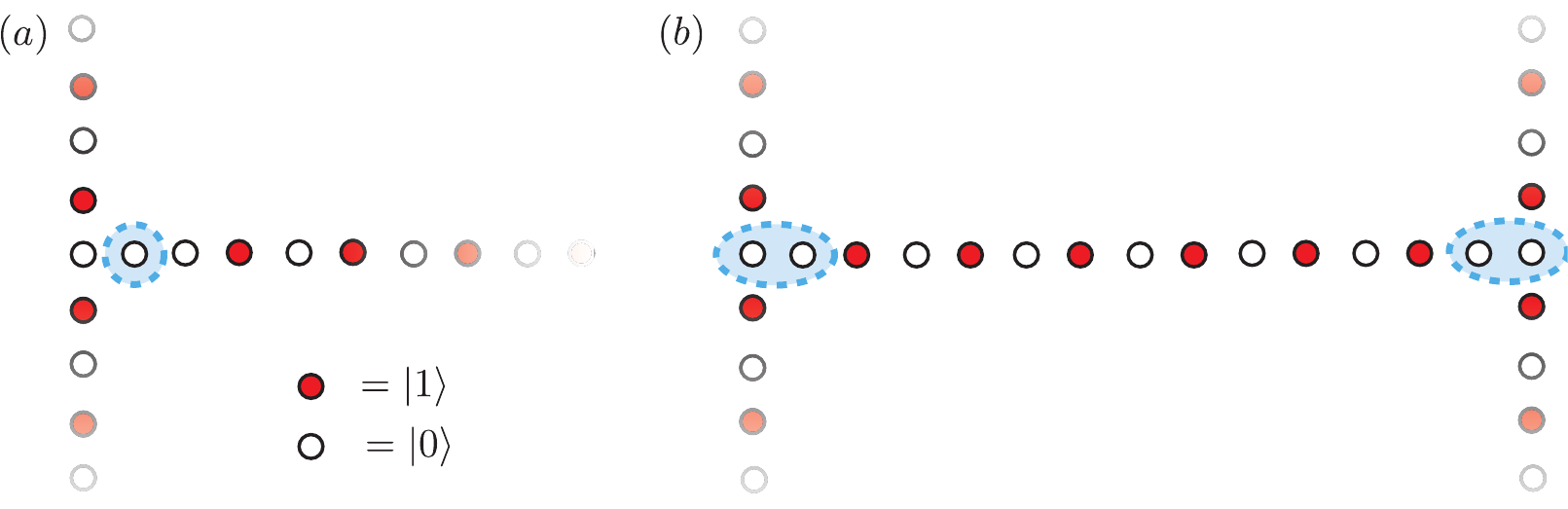}
\caption{(a) Example of a non-maximal independent set configuration. It is energetically favorable to flip the spin indicated in blue if the detuning is $\Delta_v>0.268031\times \mc{U}$ (see Sec.~\ref{Sec:MaximalIS}). $\G$ embedded on a grid. (b) Two domain walls (indicated in blue) along a straight segment. By merging the two domain walls one obtains the defect-free configuration with one more spin in state $\ket{1}$. The increase in interaction energy between these two configurations can be bounded by $0.504748\times \mc{U}$ (see Sec.~\ref{DWSS})}\label{Fig_maximal_sets}
\end{center}
\end{figure}

\subsubsection{Domain walls on straight segments}\label{DWSS}

Under these conditions, the spins on each edge $\mc{E}$ of the graph $\G$ are therefore in an ordered configuration, i.e a configuration where spins are alternating in state $\ket{0}$ and $\ket{1}$, up to so called domain walls, where two (but not more) neighboring spins are in state $\ket{0}$. From the discussion in Sec.~\ref{Sec:NPCompletenessUD}, we know that a MIS configuration on $G$ has at most one such domain wall on array of spins connecting two special vertices $i$ and $j$. In order to ensure that the ground state of the Rydberg system respects this property, the detuning has to be chosen large enough. To see this, assume that such a straight segment is in an ordered state with $n$ domain walls. When two such domain walls are combined an additional spin, $v$, can be flipped from state $\ket{0}$ to state $\ket{1}$, lowering the single-particle contribution to the total energy by $\Delta_v$. The corresponding difference in interaction energy grows with the distance of the domain walls. In particular, it is maximal, if the two domain walls are initially on opposite ends of a segment connecting two junctions [see Fig.~\ref{Fig_maximal_sets}(b)].
It can be bounded by $E_2+E_{\rm dist}$, where 
\begin{align}
{E_2} &\le
\sum_{i,j=1}^{\infty}
\medmath{ \left[\frac{4{\mc{U}} }{((2i-1)^2+(2j-1)^2)^3}-\frac{4{\mc{U}} }{((2i-1)^2+(2j)^2)^3}\right] }\no \\
&+\sum_{i=1}^\infty \medmath{\frac{{\mc{U}} }{(2i)^6}} \leq 0.490084\times {\mc{U}} .
\end{align}
This bound can be understood by looking at Fig.~\ref{Fig_maximal_sets}(b), where we obtain the zero-domain-wall configuration from the one shown with two domain walls by shifting the ordered spins on the middle bar of the H-shape (i.e., between the domain walls) one atom to the left, and then flipping a spin next to the right junction to the state $\ket{1}$.
Thus, the first sum corresponds to the increase in interactions between the middle bar of the H-shape with the sides, and the last sum is the extra interaction from the newly flipped spin.
Thus (for large enough $k$ where $E_{\rm dist}$ is negligible) it is always favorable to merge two neighboring domain walls along a straight segment if
\begin{align} \label{min_detuning}
\Delta_v\geq  \Delta_{\rm min} \equiv 0.490084\times \mc{U}.
\end{align}

In summary we find that, for $\Delta_{\rm min}<\Delta_v<\Delta_{\rm max}$ the ground state of the Rydberg Hamiltonian is ordered on all segments connecting two special vertices, with at most one domain wall per segment. This justifies the assumption in Sec.~\ref{Sec:effectiveModel}, that we can restrict the analysis of the low energy sector to only 4 configurations of the structures $B_{i,j}$. We simply choose the detuning for all those spins to be equal and in the above range, and denote it by $\Delta_B$.

\subsection{Relevant configurations and effective spins}\label{Sec:reduction}
In this subsection we focus on each special structure $A_i$ consisting of a vertex $i$ and its $2q$ neighbors on each leg individually. From Sec.~\ref{Sec:Low_energy} we know that we can restrict the discussion to states that are ordered in $A_i$ with at most one domain wall per leg. The central idea below is to choose detuning patterns for the spins in $A_i$ such that any such domain wall is energetically pushed away from spin, $i$, out of the structure $A_i$. This justifies the assumption underlying the pseudo-spin model introduced in Sec.~\ref{Sec:effectiveModel}, that only the two ordered states on $A_i$ have to be considered for the ground state. 

\subsubsection{Corners}
We first consider a corner vertex and the spins on the two legs. We show that we can choose a detuning pattern consistent with the requirement $\Delta_{\rm max}\geq \Delta_v\geq  \Delta_{\rm min}$ such that  only two configurations are relevant in the low energy sector. These are the two ordered states across the corner. We label these two states according to the state of the corner spin, $\ket{0}$ or $\ket{1}$,  by $s=0$ and $s=1$ respectively.

Assume the corner spin is in state $\ket{0}$, and on one of the legs the two spins at distance $2p-2$ and $2p-1$ from the corner are in state $\ket{0}$, forming a domain wall. The domain wall can be moved by one unit (i.e.~$p\rightarrow p+1$) by flipping the state of the spins $2p-1$ and $2p$. The interaction energy increases in this process. Its amount can be bounded by $\mc{U}F_0(p)+\mc{O}(\mc{U}/k^5)$ with 
\begin{align}\label{eq:F0}
F_0(p)&=  \sum_{i=0}^{\infty}\medmath{\left[\frac{1}{((2i+1)^2+(2p-1)^2)^3}-\frac{1}{((2i+1)^2+(2p)^2)^3}\right]}.
\end{align}
This bound can be understood as follows. Effectively, the one spin excitation is moved by one site towards the corner. While the interaction energy of this excitation with all spins on the same leg is reduced (as long as $2p<k$), and thus is upper bounded by zero, the interaction energy with all spins on the other leg increases. Since any defect on this leg would decrease this interaction energy, it is maximized if all spins on the other leg are in the perfectly ordered state, which gives the bound \eqref{eq:F0}. 

While the interaction energy increases in this process, the contribution from the single particle term to the total energy changes by $\Delta_{2p-1}-\Delta_{2p}$. Depending on the choice of the detuning this can lead to an energy gain such that it is energetically favorable to move the domain wall by one unit away from the corner spin. 
It is thus straightforward to see that the energy is minimized if the first $2q$ spins on each of the legs are in the perfectly ordered state if the corner spin is in state $\ket{0}$ and the detuning for spins satisfy 
\begin{align}\label{eq:CornerDetuing_diff0}
\Delta_{2p-1}-\Delta_{2p}\geq \mc{U} F_0(p),
\end{align}
for $p=1,2\dots, q$.

If the corner spin is in state $\ket{1}$, we find analogously that the energy is minimized if the first $2q+1$ spins on each of the leg are in the perfectly ordered if detunings satisfy
\begin{align}\label{eq:CornerDetuing_diff1}
\Delta_{2p}-\Delta_{2p+1}\geq \mc{U} F_1(p),
\end{align}
with 
\begin{align}
F_1(p)&=\sum_{i=0}^{\infty}\left[\frac{1}{((2i)^2+(2p)^2)^3}-\frac{1}{((2i)^2+(2p+1)^2)^3}\right].
\end{align} 

The conditions \eqref{eq:CornerDetuing_diff0} and \eqref{eq:CornerDetuing_diff1} can be clearly satisfied by the choice (for $p=1,\dots, q$)
\begin{align}\label{eq:CornerDet1}
\Delta_{2p}
&=\Delta_\infty+\mc{U}\sum_{i=p}^\infty [F_1(i)+F_0(i+1)],\\
\Delta_{2p-1}
&=\Delta_\infty+\mc{U}\sum_{i=p}^\infty[ F_0(i)+ F_1(i)].\label{eq:CornerDet2}
\end{align}
These sums are convergent and can be efficiently evaluated numerically. $\Delta_x$ is monotonically decreasing with $x$ and for large $x$ approaches $\Delta_\infty +\mc{O}(\mc{U}/x^5)$. 
Moreover we find that the maximum value of $\Delta_x$ in this sequence evaluates to $\Delta_1= \Delta_\infty+0.134682  \times \mc{U}$. It is therefore clearly possible to choose $\Delta_\infty$ such that the detuning along the legs of the corner are within the required range.

With the above detuning pattern and the choice $\Delta_\infty\geq \Delta_B$, the ground state configuration is necessarily such that the $4q+1$ spins forming a corner structure are in one of the two ordered states. We label these states $s=0$ and $s=1$ according to the corresponding state of the corner spin ($\ket{0}$ and $\ket{1}$).

Note that the choice \eqref{eq:CornerDet1} and \eqref{eq:CornerDet2} fixes the detuning on all spins except the detuning of the corner spin, denoted $\Delta_C$. This will allow us to tune the relative energy between the two relevant configurations, i.e.~$E_{A_i}^1-E_{A_i}^0$. 
To calculate this difference, we first define the quantity $\mc{I}_C$ as the difference in interaction energies between the two spin configurations,
\begin{align}
\mc{I}_C&=\sum_{i=1}^{q}\sum_{j=1}^q\left[\frac{{\mc{U}}}{((2i\!-\!1)^2\!+\!(2j\!-\!1)^2)^3}-\frac{{\mc{U}}}{((2i)^2\!+\!(2j)^2)^3}\right]\no\\
&-2\sum_{i=1}^{q}\frac{{\mc{U}}}{(2i)^6}= 0.0932973\times {\mc{U}}+\mc{O}(\mc{U}/q^5).
\end{align}
Similarly we calculate the difference in energy of the two configurations due to the single particle term, 
\begin{align}
\mc{D}_C&=\Delta_C+2\sum_{p=1}^q\lr{ -\Delta_{2p-1}+\Delta_{2p}}\no\\
&=\Delta_C-\mc{U}\times0.237094+\mc{O}(\mc{U}/q^5).
\end{align}
For the corner structures we thus get $E_{A_i}^1-E_{A_i}^0=-\mc{D}_C-\mc{I}_C$ which evaluates to
\begin{align}
E_{A_i}^1-E_{A_i}^0
&= -\Delta_C+0.143797 \times \mc{U}+\mc{O}(\mc{U}/q^5),
\end{align} and which can be fully tuned by the detuning of the corner spin. 

\subsubsection{Junctions}
Junctions can be treated in a similar way as corners. Again, we want to choose a detuning pattern on the legs of a junction such that it is energetically always favorable to push domain walls away from the junction, thus guaranteeing that in the ground state, a junction can only be in one of the two ordered configurations. 

Let us first consider the case where the central spin (i.e. the degree-3 vertex) is in state $\ket{0}$. We refer to the three legs by $X$, $Y$ and $Z$, where for concreteness $X$ and $Z$ are on the same gridline.
Consider a situation where two vertices on leg $X$ at a distance $2p-2$ and $2p-1$ from the central vertex are in state $\ket{0}$ and thus form a domain wall. In order to push this domain wall one unit away from the vertex, the state of vertices $2p-1$ and $2p$ has to be flipped. The interaction energy required to do so is bounded by $\mc{U}F_0(p)$. Note that this bound holds true regardless of whether (or where) on legs $Y$ and $Z$ a domain wall is present. 
Similarly, the interaction energy required to push a domain wall on leg $Y$ by one unit from the corner is bounded by $2\mc{U}F_0(p)$.

In the other case, i.e. if the central vertex is in state $\ket{1}$, similar bounds can be found. If the two vertices on leg $X$ at a distance $2p-1$ and $2p$ from the central vertex are in state $\ket{0}$, it  increases in interaction energy when the state of vertices $2p$ and $2p+1$ are flipped, i.e. when the domain wall is pushed away from the center by one unit, is bounded by $\mc{U}F_1(p)$. Again, this bound is valid independent of the configuration of the spins on legs $Y$ and $Z$. 
The interaction energy required to push a domain wall on leg $Y$ by one unit from the corner is bounded by $2\mc{U}F_1(p)$.

Therefore, the state that minimizes the energy does not contain any domain wall on the first $2q+1$ spins on each leg if the detuning pattern satisfies 
\begin{align}
\Delta_{2p-1}^{(X)}-\Delta_{2p}^{(X)}&\geq \mc{U} F_{0}(p),\no\\
\Delta_{2p-1}^{(Y)}-\Delta_{2p}^{(Y)}&\geq 2\mc{U} F_{0}(p),\no\\
\Delta_{2p}^{(X)}-\Delta_{2p+1}^{(X)}&\geq \mc{U} F_1(p),\no\\
\Delta_{2p}^{(Y)}-\Delta_{2p+1}^{(Y)}&\geq 2\mc{U} F_1(p),\no
\end{align}
for $p=1,2,\dots q$, and $\Delta_v^{(Z)}=\Delta_v^{(X)}$. Here $\Delta_v^{(\sigma)}$ denotes the detuning of the $v$-th spin on leg $\sigma$.
This can be achieved by the choice
\begin{align}
\Delta_{2p}^{(X)}&=\Delta_\infty+\mc{U}\sum_{i=p}^\infty (F_1(i)+F_0(i+1)),\no\\
\Delta_{2p-1}^{(X)}&=\Delta_\infty+\mc{U}\sum_{i=p}^\infty( F_0(i)+ F_1(i)),\no\\
\Delta_{2p}^{(Y)}&=\Delta_\infty+2\mc{U}\sum_{i=p}^\infty (F_1(i)+F_0(i+1)),\no\\
\Delta_{2p-1}^{(Y)}&=\Delta_\infty+2\mc{U}\sum_{i=p}^\infty( F_0(i)+ F_1(i)),\no
\end{align}
and $\Delta_{v}^{(Z)}=\Delta_{v}^{(X)}$.
For the choice above we find that the maximum value of $\Delta_v^{(\sigma)}$ in this sequence evaluates to $\Delta_1^{(Y)}= \Delta_\infty+0.269364  \times \mc{U}$. Thus, all the detunings on the legs of a junction $A_i$ are within the range $[\Delta_\infty,\Delta_\infty+0.269364  \times \mc{U}]$. Thus it is clearly possible to choose $\Delta_\infty$ such that all of them are in the allowed range, $\Delta_{\rm min}<\Delta_v<\Delta_{\rm max}$. With the additional choice $\Delta_\infty>\Delta_B$, the above arguments show that only the two ordered configurations are relevant for the ground state.

Analogous to the situation of corner structures, we can use the detuning of the junction spin, denoted $\Delta_J$ to tune the relative energy between the two relevant configurations.
The difference in interaction energies of the two spin configurations in this structure,  
\begin{align}
&{\mc{I}_J}=\sum_{i=1}^q\sum_{j=1}^{q}\lr{\frac{{\mc{U}}}{(2i-2+2j)^6}-\frac{{\mc{U}}}{(2i+2j)^6}}-3\sum_{i=1}^{q}\frac{{\mc{U}}}{(2i)^6}\no\\&
+2\sum_{i=1}^{q}\sum_{j=1}^q\lr{\frac{{\mc{U}}}{((2i-1)^2+(2j-1)^2)^3}-\frac{{\mc{U}}}{((2i)^2+(2j)^2)^3}},\no
\end{align}
can be evaluated to $\mc{I}_J= 0.218387\times \mc{U}+\mc{O}(\mc{U}/q^5)$.
Similarly, we calculate the difference in energy of the two configurations due to the single particle term, 
\begin{align}
\mc{D}_J
&=\Delta_J+4\mc{U}\sum_{p=1}^q\sum_{i=p}^\infty( -F_0(i)+F_0(i+1)),\no
\end{align}
which evaluates to $\mc{D}_J=\Delta_J- 0.474188\times \mc{U} +\mc{O}(\mc{U}/q^5)$.
We thus obtain the quantity $E_{A_i}^1-E_{A_i}^0=-\mc{D}_J-\mc{I}_J$ as
\begin{align}E_{A_i}^1-E_{A_i}^0&= -\Delta_J +0.255801\times \mc{U} +\mc{O}(\mc{U}/q^5).\end{align}

\subsubsection{Other special vertices}

In addition to vertices at corners and junctions, we have other special vertices. These are vertices of degree 1 (open ends), vertices on a grid point with two legs on the same grid line (straight structures), and irregular vertices. For all of these special vertices one can repeat the above analysis. 

\paragraph{Open ends.} It is straightforward to see that it always lowers the energy if a domain wall is moved away from a spin at the end of an open leg if the detuning is constant for the spins on the leg. Hence, we can naturally restrict to the pseudo-spin states corresponding to the two ordered configurations on the $2q$ spins adjacent to the spin at the end of an open leg. 
For such pseudo-spins, we therefore get 
\begin{align}
E_{A_i}^1-E_{A_i}^0&=-\Delta_O+\mc{U}\sum_{s=1}^q\frac{1}{(2s)^6}\no\\
&=-\Delta_O+ 0.015896\times \mc{U} +\mc{O}(\mc{U}/q^5),
\end{align}
 where
 $\Delta_O$ is the detuning for the spin corresponding to the vertex of degree 1. The homogenous detuning on the $2q$ spins adjacent to the latter can be chosen to be equal to $\Delta_\infty$.

\paragraph{Straight structures.}
For a regular special vertex we can guarantee that the relevant configurations are given by the two ordered states by chosing a detuning for all $4q$ spins on the leg as $\Delta_\infty>\Delta_B$. With this choice it would be energetically favorable to move a potential domain wall into the adjacent neighboring regions. 
It is easy to evaluate the energy difference 
\begin{align}\label{eq:deg2}
E_{A_i}^1-E_{A_i}^0&=-\Delta_S+\mc{U}\sum_{i=1}^{2q}\frac{1}{(2i)^6}\no\\&=-\Delta_S+0.015896 \times \mc{U}+\mc{O}(\mc{U}/q^5).
\end{align}
Here $\Delta_S$ denotes the detuning of the special vertex.

\paragraph{Irregular structures.}
Irregular vertices can be treated identically as straight structures. Since the spacing of the spins close to the irregular vertex is slightly larger than otherwise, any domain wall will be pushed away from the irregular structure naturally, if the detuning in the irregular structure is larger that $\Delta_{B}$. Thus, again, only the two ordered  configurations are relevant for the ground state. The corresponding energy difference can be numerically evaluated for every choice of $\phi$. In the large $\phi$ (and $q$) limit we obtain the same analytic expression as in \eqref{eq:deg2}.

\subsection{Effective energy\label{Sec:Effective_energy}}

In this subsection, we determine the effective detuning $\Delta_i^{\rm eff}$ of the pseudo-spins and their effective interaction energies $U_{i,j}^{\rm eff}$.

\subsubsection{Effective interactions}
The reduction to a model of effective spins in Sec.~\ref{Sec:effectiveModel}, relies on the fact that only four spin configurations are relevant in each region $B_{i,j}$ to describe the ground state. These correspond to the four possible configurations of the spins in the adjacent regions $A_i$ and $A_j$. To see this we set that  the detuning in $B_{i,j}$ to be homogeneous, $\Delta_B$. First, note, if $s_i=1$ and $s_j=0$  (or vice versa), we found in \eqref{DWSS} that the lowest energy configuration must correspond to the perfectly ordered state on $B_{i,j}$ (with $b_{i,j}$ spins in state $\ket{1}$), if $\Delta_{\rm max}>\Delta_{B}>\Delta_{\rm min}$. Any other configuration would require at least two domain walls. 
Second, if $s_i=s_j=0$ then energy can always be lowered by arranging the spins in $B_{i,j}$ in an ordered configuration with $b_{i,j}$ spins in state $\ket{1}$, and one domain wall. It is easy to see that the position of this domain wall does not change the energy up to $\mc{O}(\mc{U}/k^5)$, such that we do not need to distinguish between the different domain wall configurations.  
Finally, if $s_i=s_j=1$ the lowest energy configuration is similarly achieved by an ordered configuration with one domain wall, if $\Delta_{\rm max}>\Delta_B> \Delta_{\rm min}$. While the position of the domain wall is again irrelevant, we note that in this case only $b_{i,j}-1$ spins are in state $\ket{1}$. 

The relevant energy differences between these different relevant configurations can be readily calculated, as 
\begin{align}
E_{B_{i,j}}^{1,0}-E_{B_{i,j}}^{0,0}\! 
&= {\mc{U}}\sum_{r=1}^{\left\lceil \frac{b_{i,j}}{2}\right\rceil}\sum_{s=1}^{\left\lfloor \frac{b_{i,j}}{2}\right\rfloor}
\medmath{ \left[\frac{1}{(2r
\!+\!2s\!-\!2)^6}-\frac{1}{(2r\!+\!2s\!-\!1)^6}\right]}\no
\end{align}
If $A_i$ and $A_j$ are not connected, this is trivially zero. Else, this term becomes independent of $i$ and $j$ in the large $q$ and $k$ limit, since $b_{i,j}=\mc{O}(k)$, and evaluates to
\begin{align}
E_{B}^{1,0}-E_{B}^{0,0} = 0.0146637 \times \mc{U}+ \mc{O}(\mc{U}/k^5)
\end{align}
Analogously, $E_{B_{i,j}}^{1,1}-E_{B_{i,j}}^{1,0}$ becomes independent of $i$ and $j$, and can be calculated via
\begin{align}
&E_B^{0,0}-E_B^{1,0} = E_B^{1,1}-E_B^{1,0}-\Delta_B+  \sum_{r=1}^{\infty} \medmath{ \frac{\mc{U}}{(2r)^6}} + \mc{O}(\mc{U}/k^5)\no\\
 & = E_B^{1,1}-E_B^{1,0}-\Delta_B + 0.015896\times \mc{U} + \mc{O}(\mc{U}/k^5)
\end{align}
This gives the effective interaction between pseudo-spins $U_{i,j}^{\rm eff}\equiv U_{\rm eff}=E_B^{1,1}+E_B^{0,0}-2E_B^{1,0}$ (see Eq.~\eqref{Def_eff_interactions}) as
\begin{align}
U^{\rm eff}=\Delta_B-0.0134313\times \mc{U}+\mc{O}(\mc{U}/k^5).
\end{align}
Observe that it is solely determined by the choice of the detuning in the connecting structures, $\Delta_B$. 

\subsubsection{Effective detuning}
The effective detuning for a pseudo spin $i$ is given by $\Delta_i^{\rm eff}=E_{A_i}^{0}-E_{A_i}^1-m_i(E_B^{1,0}-E_B^{0,0})$, where $m_i$ is the degree of the special vertex, $i$ (i.e.~$m_i=2$ for corner vertices and straight structures, $m_i=3$ for junctions, and $m_i=1$ for open legs, see Eq.~\eqref{Def_eff_detuning}). We recall that the value of $E_{A_i}^{0}-E_{A_i}^1$ can be fully tuned by the choice of the detuning of spin $i$. We choose this such that we obtain a homogeneous effective detuning, $\Delta_i^{\rm eff}=\Delta^{\rm eff}$, for all four types of pseudo-spins. This is achieved by 
\begin{align}
\Delta_C&=\Delta^{\rm eff}+0.173124 \times \mc{U},\no\\
\Delta_J&=\Delta^{\rm eff}+0.299792 \times \mc{U},\no\\
\Delta_O&=\Delta^{\rm eff}+0.0305597 \times \mc{U},\no\\
\Delta_S&=\Delta^{\rm eff}+0.0452234 \times \mc{U}.\no
\end{align}
Importantly, this is compatible with the requirement $\Delta_{\rm max}\geq \Delta_v\geq \Delta_{\rm min}$, and the realization of an effective spin model with $0<\Delta^{\rm eff}<U^{\rm eff}$. 
\\
\section{Conclusion and Outlook}
Our work demonstrates  a direct connection between the many-body problem of spins interacting via van der Waals interactions and computational complexity theory. We have shown that individual control over the positions of such spins allows one to exactly encode NP-complete optimization problems. This result is obtained from a reduction from \Mis on planar graphs with maximum degree 3. Understanding this link between many-body physics problems and complexity theory is particularly important for application in quantum optimization. Our results imply that quantum optimizers based on current experimental techniques to trap and manipulate neutral atoms \cite{MISPRL} can address NP-hard optimization problems and thus have the potential to explore a potential quantum speedup \cite{Preskill:2018gt}. 

In this work we focussed mainly on the task of exactly finding the size of the MIS and a corresponding configuration. The relation between the Rydberg blockade mechanism and \Mis on unit disk graph however suggests intriguing questions related to  approximate optimization. Remarkably, polynomial time approximation algorithms exist for \Mis on unit disk graphs~\cite{Nandy:2017ij}. It would be interesting to explore if quantum approximate optimization algorithms (QAOA) for the corresponding problems \cite{Farhi:2014wk} can outperform these classical approximate methods. In addition,  it is interesting to extend the present analysis to both exact and approximate solutions of optimization programs for more general graph structures. As described in Ref.~\cite{MISPRL}, the latter can be encoded in Rydberg blockade Hamiltonian using individual addressing and multiple atomic sublevels. 

\section*{Acknowledgements} We thank H.~Bernien, E.~Farhi, A.~Harrow, A.~Keesling, W.~Lechner, H.~Levine,   A.~Omran and P. Zoller for useful discussions. This work was supported through the National Science Foundation (NSF), the Center for Ultracold Atoms, the Air Force Office of Scientific Research via the MURI, the Vannevar Bush Faculty Fellowship and DOE. H.P. is supported by the NSF through a grant for the Institute for Theoretical Atomic, Molecular, and Optical Physics at Harvard University and the Smithsonian Astrophysical Observatory. S.C.~acknowledges support from the Miller Institute for Basic Research in Science.

\bibliography{HannesBib,additional_refs}		
\end{document}